\documentclass[10pt,journal]{IEEEtran}

\usepackage{color}
\definecolor{dgreen}{rgb}{0,0.45,0}
\usepackage[colorlinks,linkcolor=dgreen,citecolor=dgreen]{hyperref}
\usepackage{stfloats}
\usepackage{amssymb}
\usepackage{amsthm}
\usepackage{graphicx}
\usepackage{float}
\usepackage{amsmath}
\usepackage{amsmath,bm}
\usepackage{color}
\usepackage{multirow}
\usepackage{mathrsfs}
\usepackage{lineno,hyperref}
\usepackage{lipsum}
\usepackage{stfloats}

\newtheorem{remark}{Remark}

\newtheorem{proposition}{Proposition}
\newtheorem{corollary}{Corollary}

\def\begcen{\begin{center}}
\def\endcen{\end{center}}

\newcommand{\col}{\mbox{col}}
\newcommand{\rank}{\mbox{rank }}

\def\phil{\phi^\mathtt{L}}
\def\dm{d_\mathcal{M}}
\def\bbr{{\mathbb R}}

\def\cals{{\cal S}}
\def\calm{{\cal M}}

\def\hal{\frac{1}{2}}

\def\L2e{{\cal L}_{2e}}

\def\rea{\mathbb{R}}

\def\diag{\mbox{diag}}

\def\et{\epsilon_t}
\def\linf{{\mathcal L}_\infty}
\def\l2{{\mathcal L}_2}
\def\l2e{{\cal L}_{2e}}

\def\rea{\mathbb{R}}

\def\diag{\mbox{diag}}

\def\begequarr{\begin{eqnarray}}
\def\endequarr{\end{eqnarray}}
\def\begequarrs{\begin{eqnarray*}}
\def\endequarrs{\end{eqnarray*}}
\def\begarr{\begin{array}}
\def\endarr{\end{array}}
\def\begequ{\begin{equation}}
\def\endequ{\end{equation}}
\def\lab{\label}
\def\begdes{\begin{description}}
\def\enddes{\end{description}}
\def\begenu{\begin{enumerate}}
\def\begite{\begin{itemize}}
\def\endite{\end{itemize}}
\def\endenu{\end{enumerate}}

\def\lef[{\left[\begin{array}}
\def\rig]{\end{array}\right]}

\def\begcen{\begin{center}}
\def\endcen{\end{center}}
\def\begrem{\begin{remark}\rm}
\def\endrem{\end{remark}}

\def\IJACSP{{\it Int. J. on Adaptive Control and Signal Processing}}

\def\TAC{{\it IEEE Trans. Automatic Control}}

\def\IJC{{\it Int. J. of Control}}

\def\SCL{{\it Systems \& Control Letters}}

\usepackage{color}

\begin{document}

\title{On State Observers for Nonlinear Systems: A New Design and a Unifying Framework}
\author{Bowen~Yi,
        Romeo~Ortega,
        and~Weidong~Zhang
\thanks{This paper is supported by the National Natural Science Foundation of China (61473183, U1509211), China Scholarship Council, the Government of Russian Fedration (074U01, GOSZADANIE 2014/190 (project 2118)), the Ministry of Education and Science of Russian Federation (14.Z50.31.0031). \emph{Corresponding author: W. Zhang}.}
\thanks{B. Yi and W. Zhang are with Department of Automation, Shanghai Jiao Tong University, Shanghai 200240, China. \tt yibowen@ymail.com, wdzhang@sjtu.edu.cn}
\thanks{R. Ortega is with LSS, CNRS-CentraleSup\'elec, Plateau du Moulon, Gif-sur-Yvette 91192, France. \tt ortega@lss.supelec.fr}
}

\markboth{Published in IEEE TAC, 2018}
{Shell \MakeLowercase{\textit{et al.}}: Bare Demo of IEEEtran.cls for IEEE Journals}

\maketitle

%
\begin{abstract}
In this paper we propose a new state observer design technique for nonlinear systems. It combines the well-known Kazantzis-Kravaris-Luenberger observer and the recently introduced parameter estimation-based observer, which become special cases of it---extending the realm of applicability of both methods. A second contribution of the paper is the proof that these designs can be recast as particular cases of immersion and invariance observers---providing in this way a unified framework for their analysis and design. Simulation results of a physical system that illustrates the superior performance of the proposed observer compared to other existing observers are presented.
\end{abstract}
%


\IEEEpeerreviewmaketitle

\section{Introduction and Problem Formulation}
\label{sec1}
%
In this paper we are interested in the design of state observers for nonlinear control systems whose dynamics is described by\footnote{All mappings in the paper are assumed smooth.}
\begin{equation}
\label{sys}
\begin{aligned}
  \dot{x} & = f(x,u) \\
  y & = h(x),
\end{aligned}
\end{equation}
where $x \in \bbr^n$ is the system state, $u \in \bbr^m$ is the control signal, and $y \in \bbr^p$ is the {\em measurable} output signal. The problem is to design a dynamical system
\begequarr
\nonumber
\dot{\chi} & = & F(y, \chi, u)\\
\hat x & = & H(y,\chi,u)
\label{dynobs}
\endequarr
with $\chi \in \bbr^{n_\chi}$, such that
\begequ
\lab{obscon}
\lim_{t\to\infty} |\hat x(t)-x(t)|=0,
\endequ
where $|\cdot|$ is the Euclidean norm. Following standard practice in observer theory we assume that   $u$ is such that the state trajectories of \eqref{sys} are bounded.

Since the publication of the seminal paper  \cite{Luenberger1966TAC}, which dealt with linear time-invariant (LTI) systems, this problem has been extensively studied in the control literature. We refer the reader to  \cite{astolfi2008book,BESbook} for a review of the literature. In this paper we are particularly interested in three recently developed observer design techniques.
\begite
\item The Kazantzis-Kravaris-Luenberger observer (KKLO) first presented in \cite{Kazantzis1998SCL} as an extension to nonlinear systems of Luenberger's observer and further developed in  \cite{Andrieu2006SIAM}.
\item Parameter estimation-based observer (PEBO) proposed in \cite{ortega2015scl}, which translates the state observation problem into an on-line  \emph{parameter estimation} problem.
\item Immersion and invariance observer (I$\&$IO), first reported in \cite{Karagiannis2008TAC} and thoroughly elaborated in  \cite{astolfi2008book}, which proposes a more general observer framework based on the generation of attractive and invariant manifolds.
\endite

The main contributions of our paper are threefold.
\begite
\item[(C1)] Propose a new observer design technique, called [KKL+PEB]O, that combines---in a seamless way---the KKLO and PEBO designs, yielding a new design applicable to a broader class of systems.
\item[(C2)]  Prove that KKLO, PEBO and [KKL+PEB]O  can be recast as particular cases of generalized I$\&$IO---providing in this way a unified framework for their analysis and design.
\item[(C3)]   Present simulation results of the well-known DC-DC \'Cuk converter that illustrate the superior performance of the proposed observer compared to other existing observer designs.
\endite

The remainder of the paper is organized as follows. Section \ref{sec2} gives some preliminaries on KKLO and PEBO.  In Section \ref{sec3} we present the new [KKL+PEB]O. The unifying framework based on immersion and invariance is given in Section \ref{sec4}. In Section \ref{sec5} we present two academic examples that illustrate the interest of the new [KKL+PEB]O and some simulation results of a physical system that compares the new observer with other observer designs. The paper is wrapped--up with concluding remarks and future research directions in Section \ref{sec6}.

%
\section{Preliminaries}
\lab{sec2}
%
In this section we briefly present simple versions of the KKLO and the PEBO that are motivating to generate the new [KKL+PEB]O in the next section.
\subsection{Kazantzis-Kravaris-Luenberger Observer}
\label{subsec21}
The KKLO design is based on the following proposition, which is a simplified version of the more general result reported in \cite{Andrieu2006SIAM,Kazantzis1998SCL}.
\begin{proposition}
\label{pro1}
\rm
Consider the system \eqref{sys} satisfying the following assumption.
\begite
\item[\bf A1]  There exist $n_\xi$ negative real numbers $\lambda_i,\;i=1,2,\dots, n_\xi$, with  $n_\xi \geq n$, and mappings
$$
\begin{aligned}
   \phi &: \rea^n \to \rea^{n_\xi}\\
    \phil &: \rea^{n_\xi} \times \rea^p \to \rea^{n}\\
    B & : \rea^p \times \rea^m \to \rea^{n_\xi}
\end{aligned}
$$
with $n_\xi \geq n - p$, satisfying the following.
\begite
\item[(i)] The KKLO partial differential equation (PDE)
\begin{equation}
\label{KKLOPDE}
    \nabla \phi^\top(x) f(x,u) = \Lambda \phi(x) + B(h(x),u),
\end{equation}
where $\nabla:=({\partial \over \partial x})^\top$ and $\Lambda:=\diag\{\lambda_i\}$.
\item[(ii)] $\phil$ is a left inverse of $\phi$, that is,
\begequ
\lab{phil}
\phil(\phi(x),h(x))=x.
\endequ
\endite
 \endite
The KKLO
\begequarr
\nonumber
    \dot{ \xi } & = & \Lambda  \xi + B(y, u)\\
    \hat{x} & = & \phil(\xi,y),
\lab{kklo}
\endequarr
ensures \eqref{obscon}.
\end{proposition}

\begin{proof}
The proof of this proposition follows immediately defining the error signal
$$
e:=\xi - \phi(x)
$$
and noting that $\dot e = \Lambda e.$
\end{proof}
\subsection{Parameter Estimation-based Observer}
\label{subsec22}
The PEBO design proposed in \cite{ortega2015scl}, although related with the KKLO, aims at formulating the state reconstruction problem as a parameter estimation problem. Towards this end, we are looking for an injection $B(h(x),u)$ and a (left invertible) mapping $\phi(x)$ that transforms the system \eqref{sys}  into\footnote{To avoid cluttering the notation, whenever clear from context, we use the same symbols to denote mappings playing similar roles in the various observers. The subindex $(\cdot)_P$ or $(\cdot)_L$ is later used to identify the PEBO or KKLO-related mappings in the [KKL+PEB]O.}
\begin{equation}
\label{dotphi}
  \dot \phi(x) = B(h(x),u).
\end{equation}
In this way, selecting (part of) the observer dynamics as
\begequ
\label{dotxi}
  \dot{\xi} = B(y,u),
\endequ
we establish, via simple integration, the key relationship
\begequ
\lab{keyrel}
 \phi(x(t)) = \xi(t) + \theta,
\endequ
where  $\theta$ is a constant vector defined as $\theta:=\phi(x(0))-\xi(0)$. It is clear that, if $\theta$ is known, we have that
$$
x=\phil(\xi+\theta,y).
$$
Hence, the remaining task is to generate an {\em estimate} for $\theta$, denoted $\hat \theta$, to obtain the observed state
\begequ
\lab{hatxpebo}
\hat x=\phil(\xi+\hat \theta,y).
\endequ
To achieve this end, we rely on the existence of the regression model
\begequ
\lab{regmod}
\begin{bmatrix}
y \\ \dot y
\end{bmatrix}
 =
 \begin{bmatrix}
h(\phil(\xi+\theta,y)) \\

\nabla h^\top(\phil(\xi+\theta,y))f(\phil(\xi+\theta,y),u)
\end{bmatrix},
\endequ
where we underscore that $y,u$ and $\xi$ are, of course, measurable.

The main result in \cite{ortega2015scl} may be summarized as follows.

\begin{proposition}
\label{pro2}\rm
  Consider the system \eqref{sys} satisfying Assumption {\bf A1} of Proposition \ref{pro1} {\em with} $\Lambda =0$ and the dynamic extension \eqref{dotxi} verifying the following assumption.
\begite
\setlength{\itemsep}{5pt}
\setlength{\parskip}{2pt}
\item[\bf A2] There exist mappings
$$
\begin{aligned}
& M: \rea^{n_\zeta} \times \rea^{n_\xi} \times \rea^p \times \rea^m \to \rea^{n_\zeta}\\
& N: \rea^{n_\zeta} \times \rea^{n_\xi} \times \rea^p \times \rea^m \to \rea^{n_\xi}
\end{aligned}
$$
such that the dynamical system
    \begin{equation}
    \label{parest0}
    \begin{aligned}
    \dot{\zeta} & = M(\zeta, \xi, y, u) \\
    \hat{\theta} & = N(\zeta, \xi, y, u),
    \end{aligned}
    \end{equation}
defines a stable, {\em consistent} parameter estimator for the regression model \eqref{regmod}, that is $\zeta$ is bounded and
\begequ
\lab{parcon}
\lim_{t\to\infty} | \hat{\theta}(t) - \theta | = 0.
\endequ
\endite
The PEBO \eqref{dotxi}, \eqref{hatxpebo}, \eqref{parest0} verifies  \eqref{obscon}.
\end{proposition}
\subsection{Remarks}
\label{subsec23}
%
\noindent{\bf R1}\quad Notice that the KKLO \eqref{kklo}, together with the dynamics of the system \eqref{sys},  admits an {\em attractive and invariant} manifold
$$
\cals:=\{ (x,\xi) \in \rea^n \times \rea^{n_\xi}\;|\; \xi = \phi(x)\},
$$
and the state is (asymptotically) reconstructed, via $\phil$, with the knowledge of $\xi$. On the other hand, the PEBO generates an {\em invariant} foliation
$$
\cals_\theta:=\{ (x,\xi) \in \rea^n \times \rea^{n_\xi}\;|\; \xi = \phi(x)+\theta,\;\theta \in \rea^{n_\xi}\}.
$$
To reconstruct the state---again via $\phil$---it is necessary to identify the leaf of the foliation where the system evolves. These observations are essential to establish the connection of these observers with the I\&IO, which also relies on the generation of an attractive and invariant manifold, defined by an invertible mapping.\\

\noindent{\bf R2}\quad Besides the additional difficulty of needing to estimate $\theta$, the main drawback of PEBO is that it relies on the open-loop integration \eqref{dotxi}, which might be a problematic operation in practice. In spite of that, PEBO  has proven instrumental in the solution of numerous physical systems problems \cite{bobtsov2015automatica,bobt2017MagLev,CHOetal,PYRetal}---some of them being unsolvable with other observer design techniques---with the open integration problem solved via the addition of ``safety nets" similar to the ones used in PID designs or adaptive control.\\

\noindent{\bf R3}\quad We underscore the fact that the PDE that needs to be solved for PEBO is {\em exactly} the one of KKLO with $\Lambda=0$, that is
$$
    \nabla \phi^\top(x) f(x,u) =  B(h(x),u).
$$
We refer the reader to \cite{pebo2017} where the generation of virtual outputs via signal injection technique of \cite{COMetal} is proposed to simplify the solution of this PDE.

%
\section{New [KKL+PEB] Observers}
\label{sec3}
%
In this section we present our first main contribution, namely, a new observer design technique that combines PEBO and KKLO. The key idea of the new observer is to split the states to be estimated in two groups, the first one estimated with a KKLO and the second one with a PEBO.
\subsection{Main result}
\label{subsec31}
The following proposition, whose proof follows {\em verbatim} from Propositions \ref{pro1} and  \ref{pro2} formalises the discussion above. For ease of presentation, and without loss of generality, we assume that the aforementioned groups are arranged one on top of the other.
\begin{proposition}
\label{pro3}
\rm
Consider the system \eqref{sys} satisfying Assumption {\bf A1} of Proposition \ref{pro1} {\em with}
\begequ
\lab{lam0}
\begin{aligned}
& \Lambda & = & \lef[{cc} \Lambda_L & {\bf 0}_{q \times (n_\xi-q)} \\  {\bf 0}_{(n_\xi -q) \times q} &  {\bf 0}_{(n_\xi -q) \times (n_\xi-q)}\rig]\; \\ & \Lambda_L: &= & \diag\{\lambda_1,\dots,\lambda_{q}\},
\end{aligned}
\endequ
where $0 \leq q \leq n_\xi,$ ${\bf 0}_{k \times j}$ is a $k \times j$ matrix of zeros, $\lambda_i < 0,\;i=1,\dots,q$. Partition the mapping $B(y,u)$ as follows
$$
B(y,u)=\lef[{c} B_L(y,u) \\ B_P(y,u) \rig],\;B_L \in \rea^{q},\;B_P\in \rea^{n_\xi -q}.
$$
The [KKL+PEB]O
\begequ
\begin{aligned}\label{pebo-eq}
        \dot{\xi }_L & = \Lambda_L{\xi }_L + B_L(y,u) \\
       \dot{\xi}_P & = B_P(y,u) \\
        \dot{\zeta} & = M_P(\zeta, \xi_P,  y, u)\\
           \hat \theta &=N_P(\zeta, \xi_P,y,u) \\
           \hat{x} & = \phil \left( \begin{bmatrix} \xi_L \\ \xi_P \end{bmatrix} +  \begin{bmatrix} \mathbf{0}_q \\  \hat \theta \end{bmatrix},y \right),
\end{aligned}
\endequ
where $\mathbf{0}_q$ is a $q$-dimensional vector of zeros, ensures \eqref{obscon} provided the mappings
$$
\begin{aligned}
& M_P: \rea^{n_\zeta} \times \rea^{n_\xi-q} \times \rea^p \times \rea^m \to \rea^{n_\zeta},\;\\
& N_P: \rea^{n_\zeta} \times \rea^{n_\xi-q} \times \rea^p \times \rea^m \to \rea^{n_\xi-q},
\end{aligned}
$$
define a consistent estimator, that is, \eqref{parcon} holds.
\end{proposition}
\subsection{Remarks}
\label{subsec32}
%
\noindent{\bf R4}\quad It is clear that Proposition \ref{pro3} contains, as particular cases, Propositions \ref{pro1} and \ref{pro2}. Indeed, the former is recovered setting $q=n_\xi$ while the latter follows with $q=0$. \\

\noindent {\bf R5}\quad The result of Proposition \ref{pro3} can be extended in several directions. For instance, it is possible to replace the PDE \eqref{KKLOPDE} by
$$
\nabla \phi^\top(x) f(x,u) = A \phi(x) + B(h(x),u),
$$
where $A$ is such that there exists a unitary matrix $P$ ensuring $A=P^\top \Lambda P$ with $\Lambda$ given in \eqref{lam0}.\footnote{With a slight modification it is also possible to consider the case of $A$ with purely imaginary eigenvalues.} Clearly, the degree of freedom provided by the inclusion of the matrix $A$ enlarges the set of solutions of the PDE. In this case, the dynamics of $(\xi_L,\xi_P)$ in the observer \eqref{pebo-eq} is  replaced by
\begequarrs
   \lef[{c} \dot \xi_L\\ \dot \xi_P \rig]  & = & \Lambda \lef[{c} \xi_L\\ \xi_P \rig]  + P B(y, u)\\
    \hat{x}& = & \phil(P^\top \lef[{c} \xi_L\\ \xi_P \rig] + \lef[{c} \mathbf{0}_q \\ \hat \theta \rig],y).
\endequarrs
\\

\noindent {\bf R6}\quad In the case of input-affine systems, {\em i.e.}, $f(x,u)=F(x)+g(x)u$, it is possible to decompose the PDE  \eqref{KKLOPDE} into two, that is,
\begequarrs
\nabla \phi^\top(x) F(x) & = & A \phi(x) + B_F(h(x))\\
\nabla \phi^\top(x) g(x) & = & B_g(h(x))
\endequarrs
and define the observer dynamics via
$$
B(y,u)  :=  B_F(y) + B_g(y)u.
$$
Explicit formulas for the solutions of these equations may be found in \cite{pebo2017}.
%
\section{I\&I Observers: An Unifying Framework}
\label{sec4}
%
In this section we show that a mild extension of the I\&IO studied in \cite{Karagiannis2008TAC,astolfi2008book} allows us to capture, as a particular case the new [KKL+PEB]O proposed in this paper---and, consequently, it also contains the KKLO and the PEBO.

\subsection{Extension of I\&I observers}
\label{subsec41}
%
The main result of the I\&IO in \cite{Karagiannis2008TAC} is extended in the following proposition by relaxing a dimension requirement imposed to some mappings in the original formulation of I\&IO---see {\bf R8} in Subsection \ref{subsec43}.

\begin{proposition}
\label{pro4}
\rm
Consider the system \eqref{sys}. Assume the existence of mappings
$$
\begin{aligned}
 \beta &: \rea^p \times \rea^{n_\chi} \to \rea^{n_z}\\
 \phi &: \rea^n \to \rea^{n_z}\\
  \phil & : \rea^{n_z} \times \rea^p \to \rea^{n}
 \end{aligned}
$$
  with $\phil(\phi(x),y) =  x$ and $n_\chi \ge n_z$, such that the following assumptions hold.
\begite
    \setlength{\itemsep}{5pt}
    \setlength{\parskip}{2pt}
    \item[\bf A4] $\rank  \nabla_\chi \beta^\top (y,\chi) = n_z$.

    \item[\bf A5] The system with state
    \begequ
    \lab{dynamic_dm}
    \begin{aligned}
     \dot{d}_{\mathcal{M}}  = &
              \nabla_y\beta^\top  \left( \nabla h^\top (x)  f(x,u)  - \nabla h^\top (\hat{x})  f(\hat{x},u) \right) \\
           &       -  \nabla \phi^\top (x)  f(x,u) +   \nabla \phi^\top(\hat{x})   f(\hat{x},u)
    \end{aligned}
    \endequ
   where
   \begequ
   \lab{hatxii}
   \begin{aligned}
     & \hat{x} & := &\phil(\phi(x)+\dm, h(x))\\
     & \dm &:= &\beta(y,\chi) - \phi(x),
   \end{aligned}
   \endequ
   has an asymptotically stable equilibrium $\dm=0$.
\endite
  The I\&IO
  \begequarr
    \lab{function_alpha}
    \dot{ \chi } & = &
     -  [\nabla_\chi \beta^{\top}]^\dagger
       ( \nabla_y \beta^\top \nabla h^\top(\hat{x})  -  \nabla \phi^\top (\hat{x})  )  f(\hat{x},u) \;\;\; \\
       & &
       +  (  I -  [\nabla_\chi \beta^{\top}]^\dagger  \nabla_\chi \beta^\top  )
          Q (y,\chi,u) \nonumber \\
    \hat{x}& = & \phil(\beta(y,\chi),y).
    \lab{I&Io}
    \endequarr
with  $[\cdot]^\dagger$ the pseudoinverse and  $Q: \rea^p \times\rea^{n_\chi} \times \rea^m \to \rea^{n_\chi \times n_\chi}$ an arbitrary mapping, verifies \eqref{obscon}.
\end{proposition}

\begin{proof}
The dynamics of off-the-manifold coordinate $\dm$ is
$$
\dot{d}_{\mathcal{M}} = \nabla_y \beta^\top \nabla h^\top(x) f(x,u) +   \nabla_\chi \beta^\top \dot{\chi}    -    \nabla \phi^\top(x) f(x,u).
$$
Replacing the dynamics of $\dot{\chi}$ in \eqref{function_alpha}, we get \eqref{dynamic_dm}. According to Assumption {\bf A5}, we have
$$
\lim_{t\to \infty} \dm(t) =0.
$$
Replacing this limit in \eqref{hatxii}  and recalling that  $\phil(\phi(x),y) =  x$ ensures $\lim_{t\to\infty} |\hat{x}(t) - x(t)|=0.$
\end{proof}

\subsection{KKL+PEB observers: An I\&I interpretation}
\label{subsec42}
%
In this section we will show that if the system admits a [KKL+PEB]O it also admits an I\&IO. To unify the notation we define
$$
\chi  :=  \begin{bmatrix}\xi_L \\ \xi_P \\ \zeta \end{bmatrix}, \;
\xi  :=  \begin{bmatrix}\xi_L \\ \xi_P \end{bmatrix},\;
\upsilon  :=  \begin{bmatrix}\xi_P  \\ \zeta \end{bmatrix},
$$
introduce the partitions
$$
\begin{aligned}
\phi(x) & := \begin{bmatrix}\phi_L(x)  \\ \phi_P(x) \end{bmatrix} \\
M& := \begin{bmatrix} \mathbf{0}_q \\  M_P(\zeta, \xi_P,y,u) ) \end{bmatrix} \\
N & :=\begin{bmatrix}\mathbf{0}_q,N(\zeta, \xi_P,y,u)\end{bmatrix}
\end{aligned}
$$
and define the mapping
\begequ
\label{upsilon}
\begin{aligned}
  \Upsilon(x,\chi,u):= & \nabla_y\beta^\top  \left( \nabla h^\top (x)  f(x,u)  - \nabla h^\top (\hat{x})  f(\hat{x},u) \right)\\
                  &  -  \nabla \phi^\top (x)  f(x,u) +   \nabla \phi^\top(\hat{x})   f(\hat{x},u),
\end{aligned}
\endequ
that, according to \eqref{dynamic_dm},  defines the dynamics of the off-the-manifold coordinate ${d}_\mathcal{M}$.

\begin{proposition}
\label{prop5}
\rm
Assume the system \eqref{sys} admits a [KKL+PEB]O \eqref{pebo-eq} with
$$
\rank \big(\big[I_{n_{{\xi}}}+ \nabla_{{\xi}} N^\top ~\big|~\nabla_{\zeta}N^\top\big]\big) = n_{\xi},
$$
and $\{(x,\xi,\zeta)|\theta - \hat{\theta} =0 \}$ is invariant. Then, the system admits an I\&IO \eqref{function_alpha}-\eqref{I&Io} with the mappings selected as \eqref{mapping3-1}, \eqref{mapping3-2} and \eqref{mapping3-3}.
\begin{align}
&\beta  = \xi + N(\zeta,\xi_P,h(x)) \label{mapping3-1}  \\
&\Upsilon  = \nonumber\\
& ~~ \begin{bmatrix}
                               - \Lambda_L \phi_L(x) + \Lambda_L \xi_L \\
                                 \mathbf{0}_{n_\xi - q} \\
                                \nabla_{\zeta}\beta^\top  \big( M_P(\zeta, \phi_P(\hat{x}),h(x),u)  - M_P(\zeta, \phi_P(x),h(x),u)  \big)
                              \end{bmatrix}
                                 \label{mapping3-2}   \\
  &     Q  = \mathbf{0}. \label{mapping3-3}
\end{align}
If the measurable output signals {\em are} partial states, the [KKL+PEB]O \eqref{pebo-eq} \emph{exactly} coincides with the I\&IO \eqref{function_alpha}-\eqref{I&Io}, and
$$
\mbox{I\&IO PDE \eqref{upsilon}, \eqref{mapping3-2}}\\
\quad \Longrightarrow \quad
\mbox{[KKL+PEB]O PDE \eqref{KKLOPDE}}.
$$
\end{proposition}

\begin{proof}
For the sake of clarity, we assume $N_P$ is independent of $u$ to avoid further complicating the notation. Before the proof, we present the following two useful facts.
\begin{enumerate}
  \item[\bf F1] If the output signals are partial states, \emph{i.e.}, $x :={\rm col}(\mathbf{x}_1, \mathbf{x}_2)$ and $y = \mathbf{x}_2$ without loss of generality, we have
$$
h(\hat{x})  = h(\phil (\xi + N(\zeta,\xi_P,y),y ))
                     = h( {\rm col}(\hat{\mathbf{x}}_1, \mathbf{x}_2))
                     = h(x).
$$
\item[\bf F2] When $\xi_P = \phi_P(x)$, we have
$$
{d\over dt} N_P( \zeta,\phi_P(x), y)
 = \mathbf{0}_{n_{\xi} -q},
$$
   thus yielding the following identity.
\begequ
\lab{Eq-B2}
\begin{aligned}
   \nabla_{\xi_P} N_P^\top \nabla \phi_P^\top(x) f(x,u) &+ \nabla_{\zeta} N_P^\top M_P(\zeta,\phi_P(x),y,u) \\
&  + \nabla_y N_P^\top  \nabla h^\top(x) f(x,u)   = \mathbf{0} .
\end{aligned}
\endequ
\end{enumerate}

According to \eqref{mapping3-1}, Assumption {\bf A4} is obviously satisfied. The reminding of the proof is divided into two parts: 1) the selected mappings yield the dynamics \eqref{function_alpha}-\eqref{I&Io}, having the same structure as \eqref{pebo-eq} in [KKL+PEB]O; 2) these mappings are solutions of I\&IO PDE.

\textbf{1)} The dynamics of $\dot{\chi}$ in \eqref{function_alpha} has the term $[\nabla_\chi \beta^\top]^{\dagger}   ( \nabla_y\beta^\top \nabla h^\top  - \nabla \phi ^\top  )  f(\hat{x},u) $, in which
%
\begequ
\lab{relation1}
\begin{aligned}
&  \nabla_y \beta^\top \nabla h^\top(\hat{x})    f(\hat{x},u)  -  \nabla \phi^\top (\hat{x})  f(\hat{x},u) \\
= &
        \begin{bmatrix}
        \mathbf{0}_q - \nabla \phi_L^\top(\hat{x})\\
        \nabla_y N_P^\top \nabla h^\top(\hat{x}) - \nabla \phi_P^\top(\hat{x})
          \end{bmatrix}
          f(\hat{x}, u)
\end{aligned}
\endequ
%
We analyze the above equation in two parts, \emph{i.e.}, $ - \nabla \phi_L^\top  f$ and $(\nabla_y N_P^\top \nabla h^\top - \nabla \phi_P^\top) f$. For the first part, the existence of a [KKL+PEB] observer yields the PDE
$$
\Lambda_L\phi_L(x) + B_L(y,u) = \nabla \phi_L^\top (x)  f(x,u).
$$
Substituting $x$ by $\hat{x} = \phi^{\mathtt{L}}(\xi_L)$, we have
\begequ
\lab{relation3}
  \nabla \phi_L^\top (\hat{x})  f(\hat{x}, u) = \Lambda_L \xi_L + B_L(y,u).
\endequ
The second partition of \eqref{relation1} verifies the relation below.
\begequ
\lab{relation4}
\begin{aligned}
&  (\nabla_y N_P^\top \nabla h^\top (\hat{x})- \nabla \phi_P^\top(\hat{x})) f(\hat{x},u)\\
 \overset{\eqref{Eq-B2}}{=} & - ( \nabla_{\xi_P} N_P^\top (\zeta,\phi_P(\hat{x}), h(\hat{x})) + I)  \nabla \phi_P^\top (\hat{x})  f(\hat{x},u) \\
                                          &              - \nabla_{\zeta}  N_P^\top (\zeta,\phi_P(\hat{x}), h(\hat{x})) M_P(\zeta,\phi_P(\hat{x}),h(\hat{x}),u) \\
 = & - \nabla_{\upsilon} \beta^\top \begin{bmatrix} B_P( h(\hat{x}),u) \\  M_P(\zeta,\xi_P ,h(\hat{x}),u) \end{bmatrix}.
\end{aligned}
\endequ

Combining \eqref{relation1}-\eqref{relation4}, we get the mapping $\alpha(\cdot)$ in I\&IO is
$$
\dot{\chi} = \alpha(y,\chi,u) = \begin{bmatrix}
                            \Lambda_L {\xi }_L + B_L(y,u)\\
                            B_P( h(\hat{x}),u) \\
                             M_P(\zeta, \xi_P ,h(\hat{x}),u)
                             \end{bmatrix},
$$
showing that \emph{if} the I\&IO PDE has a solution, then the I\&IO \emph{asymptotically} coincides the [KKL+PEB]O. Furthermore, if the measurable output signals are partial states, due to fact {\bf F1}, the I\&IO \emph{exactly} coincides with the [KKL+PEB]O.

\textbf{2)} We check the solution existence of I\&IO PDE. The first partition of I\&IO PDE is verified as follows.
\begequ
\lab{relation2}
\begin{aligned}
& \quad  \Lambda_L\phi_L(x) + B_L(y,u) = \nabla \phi_L^\top (x)  f(x,u). \\
\Rightarrow &\quad - \Lambda_L\phi_L(x) - B_L (y,u)  +  \Lambda_L\xi_L + B_L(y,u)  \\
& \quad= -\nabla \phi_L^\top (x) f(x,u) + \nabla \phi_L^\top (\hat{x})  f(\hat{x},u)   \\
 \Leftrightarrow
& \quad \Lambda_L\xi_L  - \Lambda_L\phi_L(x) =
 \nabla \phi_L^\top (\hat{x})  f(\hat{x},u)  -\nabla \phi_L^\top (x)  f(x,u)
\end{aligned}
\endequ
For the second partition of the I\&IO PDE, the identity \eqref{Eq-B2} yields \eqref{relation5}. Combining \eqref{relation2}-\eqref{relation5}, it shows that the selecting mappings are solutions of I\&IO PDE.

The invariant manifold of [KKL+PEB]O is
$$
\calm = \left\{ (x,\chi) \in \rea^n \times \rea^{n_\chi} \big|  \;
                           \xi   + N(\zeta,\xi_P,y)   =  \phi(x)  \right\}.
$$
The dynamics of off-the-manifold coordinate is $\dot{d}_{\mathcal{M}} =  \Upsilon(x,\chi,u)$, whose convergence is guaranteed by the consistent identification Assumption {\bf A2} and the fact that the matrix $\Lambda_L$ is Hurwitz. This completes the proof.
\end{proof}
\begin{figure*}[b]
\hrulefill
\normalsize
\begequ
\lab{relation5}
\begin{aligned}
    &
   -\nabla_{\upsilon} \beta^\top(h(x),\chi) \begin{bmatrix} B_P(h(x),u) \\ M_P(\zeta, \phi_P(x),h(x),u) \end{bmatrix}
                        + \nabla_{\upsilon} \beta^\top(h(\hat{x}),\chi) \begin{bmatrix} B_P(h(\hat{x}),u) \\ M_P(\zeta, \phi_P(\hat{x}),h(\hat{x}),u) \end{bmatrix} \\
 =  & -\nabla_{\upsilon} \beta^\top(h(x),\chi) \begin{bmatrix} \nabla \phi_P^\top(x) f(x,u) \\ M_P(\zeta, \phi_P(x),h(x),u) \end{bmatrix}
                        + \nabla_{\upsilon} \beta^\top(h(\hat{x}),\chi) \begin{bmatrix} \nabla \phi_P^\top(\hat{x}) f(\hat{x},u) \\ M_P(\zeta, \phi_P(\hat{x}),h(\hat{x}),u) \end{bmatrix}  \\
  \Longleftrightarrow  &\quad \quad  \nabla \phi^\top_P(x) f(x,u) = B_P(h(x),u).
\end{aligned}
\endequ
\hrulefill
\vspace*{4pt}
\end{figure*}

We are in position to present the main result of this paper---the unified observer framework captured by I\&IO.
\begin{corollary}\rm
  For the nonlinear system \eqref{sys}, a [KKL+PEB]O implies the existence of I\&I observers. Moreover, the
following ``set" relationship holds:
 \begequ
 \lab{setdia}
 \left.
\begin{aligned}
\mbox{PEBO}\\
\mbox{KKLO}
\end{aligned}
\right\}
\subset
\mbox{[KKL+PEB]O }
\subset
\mbox{I\&IO}.
\endequ
\end{corollary}

\subsection{Remarks}
\label{subsec43}
%
\noindent{\bf R7}\quad As discussed above, an I\&IO generates the  invariant manifold
$$
\calm = \left\{ (x,\chi) \in \rea^n \times \rea^{n_\chi} \big|  \; \beta(h(x),\chi)  =  \phi(x)  \right\},
$$
which is made attractive ensuring---via Assumption {\bf A5}---that the zero equilibrium of the dynamics of the off-the-manifold coordinate  \eqref{dynamic_dm}, which may be written as
{
$$
 \dot{d}_{\mathcal{M}}  = \Upsilon(x,\chi,u)  =:  \Upsilon_0(x,\dm,u)
$$}
has an asymptotically stable equilibrium at the origin.\\

\noindent {\bf R8}\quad In the I\&IO proposed in  \cite{Karagiannis2008TAC,astolfi2008book} we fix $ n_\chi  = n_z \le n $. In this case, \eqref{function_alpha} reduces to
$$
  \dot{\chi}  =
     - [\nabla_\chi \beta^{-1}]^\top
       (  \nabla_y\beta^\top \nabla h^\top(\hat{x})      -  \nabla \phi^\top (\hat{x})  ) f(\hat{x},u),
$$
and {\bf A4} is equivalent to requiring that $\nabla_\chi\beta$ is a non-singular square matrix. For PEBO and [KKL+PEB]O, the dimensions of their corresponding invariant manifolds are less than the dimensions of dynamic extensions. Hence, we generalise I\&IO removing the requirement $n_\xi = n_z$ and using the pseudoinverse.\\

\noindent{\bf R9}\quad KKLOs and PEBOs are specific cases of [KKL+PEB]Os, making them covered by I\&IO framework. More specifically the following statements hold.
\begin{itemize}
  \item the KKLO \eqref{kklo} coincides with the I\&IO \eqref{function_alpha}-\eqref{I&Io} with mappings selecting
      $$
      \begin{aligned}
      n_{\chi} & = n_{z} = q \\
      \beta(y,\chi) & = \xi = \chi \\
      \Upsilon(x,\chi,u)  & = -\Lambda\phi(x) + \Lambda\chi
      \end{aligned}
      $$
      with \emph{any} mapping $Q(y,\chi,u)$. The KKLO PDE sacrifices the freedom for $\Upsilon$ by fixing $\Upsilon = - \Lambda\phi(x) + \Lambda\chi$.
  \item If the measurable output signals are partial states, the PEBO coincides with the I\&IO with mappings \footnote{Here we also assume $N$ is independent of $u$ for the sake of clarity.}
      $$
      \begin{aligned}
     & \chi  =\col(\xi,\zeta) \\
     & \beta = \xi + N(\zeta,\xi,y) \\
      &\Upsilon =   \nabla_\chi\beta^\top
      \begin{bmatrix}
        \mathbf{0} \\
        M(\zeta, \phi(\hat{x}),h(\hat{x}),u) - M(\zeta, \phi(x),h(x),u)
      \end{bmatrix}
      \end{aligned}
      $$
      and $Q=0$, where $ \hat{x}  = \phil(\beta(y,\chi),y)$.
\end{itemize}
~\\
\noindent{\bf R10}\quad Notice that the condition that ``the measurable output signals {\em are} partial states" in Proposition \ref{prop5} is done without loss of generality because it is always possible to do a change of coordinates to verify it.
%

\section{Examples}
\label{sec5}
%
\subsection{Proving the interest of the new observer}
\lab{subsec51}
In this section, an academic example for which neither KKLO nor PEBO are applicable, but it is solvable via our new [KKL+PEB]O design.
\begin{proposition}
\lab{pro6}
\rm
Consider the system
\begequ
\label{num_exmp}
\begin{aligned}
\dot{x}_1   & = - x_1^3 + e^{x_3}\\
\dot{x}_2  &  = - x_2 + x_1^2 + \sin x_1\\
\dot{x}_3  & =  (x_1^2 +1)^{-1} + x_1 u \\
y  &  = x_1,
\end{aligned}
\endequ
The following facts hold.
\begite
\item[{\bf F3}] The system {\em does not admit} a KKLO nor a PEBO.
\item[{\bf F4}]  The system {\em admits} a [KKL+PEB]O, namely,
\begin{align}
 \dot{\xi}_1 & = - \xi_1 + y^2 + \sin y \label{dyn_ext_5.1-1} \\
 \dot{\xi}_2 & = u y + ( y^2 +1 )^{-1} \label{dyn_ext_5.1-2}\\
  \dot{ \hat{\Theta}} & = \gamma \psi (Y - \psi \hat{\Theta}) \label{exp_estimator5.1-1} \\
  \hat{x}_2 & = \xi_1 \\
\hat{x}_3 & = \xi_2 +  \ln \hat{\Theta}.
\end{align}
where $\gamma >0$ is an adaptation gain and $Y,\psi$ are obtained via LTI filtering as
\begequarr
\nonumber
Y & = & {\alpha p \over p+\alpha}\big[y\big] + {\alpha \over p +\alpha} \big[y^3\big] \\
\psi & = & {\alpha \over p+\alpha}\big[e^{\xi_2}\big],\;\psi(0)>0
\lab{ypsi}
\endequarr
with $p:={d \over dt}$ and $\alpha>0$, is a  [KKL+PEB]O that ensures
$$
\lim_{t\to\infty} |\hat x_i(t) - x_i(t)|=0,\;i=2,3.
$$
\endite
\end{proposition}
\begin{proof}
\noindent [Proof of {\bf  F3}] KKLO requires $\phi(x)$ to be injective. To guarantee this property  at least one of its three components should depend on $x_3$. Suppose, without loss of generality, that $\phi_2(x)$ depends on $x_3$. Define the three-dimensional vector $\rho$ as
\begequ
\lab{rho}
\rho(x):=\nabla\phi_2(x).
\endequ
From the PDE \eqref{KKLOPDE} we have
\begequ
\lab{lhs}
\rho^\top(x) \begin{pmatrix} -x_1^3 + e^{x_3} \\ -x_2 + x_1^2 + \sin x_1 \\ x_1 u + ( 1+x_1^2)^{-1}   \end{pmatrix} =-\lambda_2 \phi_2(x) + B_2(x_1,u)
\endequ
Since $\phi_2(x)$ dependends of $x_3$, we have $\rho_3(x) \neq 0$. The left hand side term of \eqref{lhs} dependent on $u$ is $\rho_3(x) x_1u$, while the one in the right hand side is $B_2(x_1,u)$, from which we conclude that $\rho_3(x)$ {\em only depends} on $x_1$, that is $\rho_3(x) = \rho_3(x_1)$. From Poincare's lemma we have that \eqref{rho} holds if and only if the Jacobian $\nabla \rho(x)$ is a {\em symmetric} matrix.  Applying  this condition to the $(1,3)$ element of the Jacobian we conclude that $\rho_1(x)$ should satisfy
$$
\rho_1(x) :=\rho'_3(x_1) x_3 + L(x_1,x_2),
$$
with $L(x_1,x_2)$ to be determined and $(\cdot)'$ the derivative with respect to its argument. From the $(2,3)$ element we also conclude that $\rho_2(x)$ is independent of $x_3$, that is $\rho_2(x) : = \rho_2(x_1,x_2)$.

The terms in the left-hand side of \eqref{lhs}  dependent on $x_3$ are $-\rho'_3(x_1) x_3 x_1^3  + \rho'_3(x_1) x_3 e^{x_3} + L(x_1,x_2) e^{x_3}$, while the one on the right-hand side is $-\lambda_2\rho_3(x_1) x_3$. Thus we conclude that $L(x_1,x_2) =0$ and
$$
-\lambda_2\rho_3(x_1) x_3 = -\rho'_3(x_1) x_3 x_1^3  + \rho'_3(x_1) x_3 e^{x_3}.
$$
Hence
$$
-\lambda_2\rho_3(x_1) = \rho'_3(x_1)(-  x_1^3  +   e^{x_3}),
$$
whose only solution is $\rho_3 =0$, which contradicts with the fact that $\rho_3 \neq 0$, due to $\lambda_2 \neq 0$ in the KKLO. Therefore, it shows that the system \eqref{num_exmp} does not admit a KKLO.

For the injectivity of $\phi(x)$ in PEBO, assume that $\phi_2(x)$ depends on $x_2$ and $\rho_2(x)\neq 0$. It follows from the argument above that, for the PEBO with $\lambda_2 =0$ we have $\rho_3'(x_1)=0$, yielding $\rho_3(x_1) = c$ and $\rho_1(x)=0$ with a constant $c$. From the (1,2) and (2,1) elements of the Jacobian matrix $\nabla \rho(x)$ we conclude that $\nabla_{x_1}\rho_2=0$ and $\rho_2(x_1,x_2):=\rho_2(x_2)$. Then in terms of \eqref{lhs}, we have
$$
-\rho_2(x_2)x_2 + \rho_2(x_2) x_1^2 + \rho_2(x_2) \sin x_1 = B_2(x_1,u).
$$
Since the first term in the left hand side does not depend on $x_1$, we conclude that $\rho(x_2) =0$ leading to a contradiction. Thus the given system does not admit a PEBO.
\\

\noindent [Proof of {\bf  F4}]  The [KKL+PEB]O PDE  \eqref{KKLOPDE} with \eqref{lam0} has a solution as $q=1$, $\lambda_1=-1$, $\phi(x) = \col(x_2,x_3)$ and
$$
B(h(x),u)=
\begin{bmatrix}
 x_1^2 + \sin x_1 \\
  x_1 u + {1 \over x_1^2 +1}
\end{bmatrix}.
$$
Thus the observer dynamics is given by \eqref{dyn_ext_5.1-1} and \eqref{dyn_ext_5.1-2}. From which we conclude that
$$
\lim_{t\to \infty} |\xi_1(t) - x_2(t)|=0, \; x_3(t) =  \xi_2(t) + \theta,
$$
with the constant $\theta$ to be estimated. Noticing the following relationship $\dot{x}_1 = -x_1^3 + e^{x_3}$, we can formulate a  linear regression model for the estimation of $\theta$ of the form
$$
Y = \psi \Theta + \epsilon_t
 $$
where $\Theta:= \exp(\theta)$, $Y$ and $\psi$ are defined in \eqref{ypsi} and $\epsilon_t$ is a (generic) exponentially decaying term that, without loss of generality, we neglect in the sequel.\footnote{See Remark 3 in \cite{ARAetaltac} where the effect of these term is rigorously analyzed.}  Finally, the choice of initial condition for $\psi$ ensures that $\psi(t)$ is not square integrable, thus
$$
\dot{\tilde{\Theta}} = - \gamma \psi^2 \tilde{\Theta}
$$
with $\tilde{\Theta}:= \hat{\Theta} - \Theta$ ensures $\lim_{t\to\infty} \hat \Theta(t)=\Theta$ and consequently \eqref{obscon} is guaranteed.
\end{proof}
%
%

%
\subsection{A class of nonlinear systems for [KKL+PEB]O}
\label{subsec52}
%
We identify now a class of systems whose states can be reconstructed with the proposed [KKL+PEB]O.

\begin{proposition}
\lab{pro7}
\rm
Consider systems of the form
\begin{align}
\label{exa2}
\dot{\mathbf{x}}_1 & = \mathbf{f}_1( \mathbf{x}_1,\mathbf{x}_{2},\mathbf{x}_{3},u) + S({x},u)   \nonumber\\
\dot{\mathbf{x}}_2 & = {\mathbf A}_2\mathbf{x}_2 + \mathbf{f}_2 (\mathbf{x}_1, u)   \nonumber\\
\dot{\mathbf{x}}_{3} & = {\mathbf A}_3 \mathbf{x}_{3} + \mathbf{f}_{3} (\mathbf{x}_1, \mathbf{x}_{2},u) \\
\dot{\mathbf{x}}_{4} & = \mathbf{f}_4 (\mathbf{x}_1,\mathbf{x}_{2}, \mathbf{x}_{3},u)  \nonumber\\
 {y} & = \mathbf{x}_1,  \nonumber
\end{align}
where  ${x}:=\col(\mathbf{x}_{1},\mathbf{x}_{2}, \mathbf{x}_3,\mathbf{x}_{4})$, with $\mathbf{x}_i \in \rea^{n_i},\;i=1,\dots,4$, and $u \in \rea^{m}$, verifying the following assumptions.
\begite
\item[(i)] There exists $1 \leq k \leq n_1$ such that the corresponding element of the vector $S$ satisfies
\begequ
\lab{sk}
S_k({x},u)=b^\top (\mathbf{x}_1,\mathbf{x}_{2},\mathbf{x}_{3},u) \mathbf{x}_{4}
\endequ
for some mapping $b:\rea^{n_1} \times \rea^{n_2} \times \rea^{n_3} \times \rea^m \to \rea^{n_4}$.
\item[(ii)] The matrices $ {\mathbf A}_2$ and $ {\mathbf A}_3$ are Hurwitz.
\item[(iii)]  The control input guarantees that the trajectories of \eqref{exa2} are bounded and the following persistency of excitation condition is verified
\begequ
\lab{pe}
\int_t^{t+T} b(s)
 b^\top(s)ds \geq \delta I_{n_4},
\endequ
for all $t \geq 0$ and some $\delta,T>0$.
\endite
The system admits a [KKL+PEB]O  of the form
\begin{align}
\lab{x2}
\dot{\hat {\mathbf{x}}}_2 & =  {\mathbf A}_2 \hat {\mathbf{x}}_2 + \mathbf{f}_2 (y, u) \\
\lab{x3}
\dot{\hat {\mathbf{x}}}_{3} & =  {\mathbf A}_{3} \hat {\mathbf{x}}_{3} + \mathbf{f}_{3} (y,\hat {\mathbf{x}}_{2},u) \\
\lab{dotxi0}
\dot{\mathbf{\xi}} & = \mathbf{f}_4 (y,\hat{\mathbf{x}}_2,\hat{\mathbf{x}}_{3}, u)   \\
\lab{x4}
\hat{\mathbf{x}}_4 & = \mathbf{\xi} + \hat \theta,
\end{align}
with parameter estimator
\begequ
\lab{paa}
\dot{ \hat{\theta}}  =   \Gamma \hat \psi (\hat {Y} - \hat {\psi}^\top\hat{\theta})
\endequ
where
$$
\begin{aligned}
\hat {Y} & :=  {\alpha p \over p + \alpha} \big[{y}_k\big] -  {\alpha  \over p + \alpha}\big[{f}_{1,k}(y,\hat{ \mathbf{x}}_{2},\hat{ \mathbf{x}}_{3},u) \big]  \\
& \quad \; -  {\alpha  \over p + \alpha}\big[b^\top (y,\hat{ \mathbf{x}}_{2},\hat{ \mathbf{x}}_{3},u)\xi \big], \\
\hat \psi &  :=  {\alpha  \over p + \alpha}\big[b (y,\hat{ \mathbf{x}}_{2},\hat{ \mathbf{x}}_{3},u) \big]
\end{aligned}
$$
with ${f}_{1,k}$ the $k$-th element of the vector $\mathbf{f}_1$.
\end{proposition}

\begin{proof}
We first prove boundedness of $ \hat{ \mathbf{x} }_2$ and $ \hat{ \mathbf{x} }_3$. Due to the assumption (iii), we have that $\mathbf{f}_2 (y, u)  \in \linf$. Hence, from \eqref{x2} and the fact that $A_2$ is Hurwitz, we have that $\hat {\mathbf{x}}_2 \in \linf$. Proceeding in the same way with \eqref{x3} we conclude that  $ \hat{ \mathbf{x} }_3  \in \linf$.

Now, we prove that the observation errors  $\tilde{x}_i(t) := \hat{x}_i(t) - x_i(t)$ ($i=1,\ldots,3$) converge to zero exponentially fast. It is obvious that
$$
\dot{\tilde{ \mathbf x}}_2 = {\mathbf A}_2 \tilde{\mathbf x}_2,
$$
and $\lim_{t \to \infty} \tilde{ \mathbf x}_2= 0$ (exp.). Similarly,
$$
\dot{\tilde{ \mathbf x}}_3 = {\mathbf A}_3 \tilde{\mathbf x}_3 + {\mathbf f}_3 (y, \hat{\mathbf x}_2, u) - {\mathbf f}_3 (y, \mathbf{x}_2, u).
$$
Consider the function $V(\tilde{ \mathbf x}_3):=\hal \tilde{ \mathbf x}^\top_3P\tilde{ \mathbf x}_3$, with $P>0$ the solution of $P {\mathbf A}_3+ {\mathbf A}_3^\top P = - \ell_1 I <0$. Its time derivative satisfies
\begequarr
\nonumber
\dot V  & = & -\ell_1 |\tilde{ \mathbf x}_3|^2 + 2 \tilde{ \mathbf x}^\top_3 P [ {\mathbf f}_3 (y, \hat{\mathbf x}_2, u) - {\mathbf f}_3 (y, \mathbf{x}_2, u) ]\\
 & \leq & -\ell_1 |\tilde{ \mathbf x}_3|^2 + \ell_2 |\tilde{ \mathbf x}_3| |\tilde{\mathbf x}_2|,
\lab{dotv}
\endequarr
where we have used the fact that the boundedness of all the arguments of ${\mathbf f}_3(\cdot,\cdot,\cdot)$ ensures
$$
| {\mathbf f}_3 (y, \hat{\mathbf x}_2, u) - {\mathbf f}_3 (y, \mathbf{x}_2, u)| \leq  \ell_3 |\tilde{\mathbf x}_2|
$$
for some $\ell_3>0$. From \eqref{dotv}, the comparison lemma and the fact that  $ \tilde{ \mathbf{x} }_3  \in \linf$ and $\tilde{ \mathbf x}_2(t) \to 0$ (exp.), we conclude that $\tilde{ \mathbf x}_3(t) \to 0$ (exp.).

It only remains to prove that $\tilde {\mathbf x}_4:= \hat {\mathbf x}_4 -  {\mathbf x}_4$ also converges to zero. Consider \eqref{dotxi0} and define $\tilde {\mathbf \xi}: = \xi - \mathbf{x}_{4}$, which satisfies
$$
\dot{\tilde{\mathbf \xi}} =  \mathbf{f}_{4} (y, \hat {\mathbf{x}}_{2},  \hat {\mathbf{x}}_{3},u) -  \mathbf{f}_{4} (y,  {\mathbf{x}}_{2},  {\mathbf{x}}_{3},u).
$$
Hence,
\begequarrs
\tilde \xi(t) & = & \int_0^t [ \mathbf{f}_{4} (y(s), \hat {\mathbf{x}}_{2}(s),  \hat {\mathbf{x}}_{3}(s),u(s)) \\
                    && \quad\quad -  \mathbf{f}_{4} (y(s),  {\mathbf{x}}_{2}(s),  {\mathbf{x}}_{3}(s),u(s))]ds + \tilde \xi(0) \\
                & \leq &  \ell_4 \int_0^t \left| \lef[{c}  \tilde {\mathbf{x}}_{2}(s) \\  \tilde {\mathbf{x}}_{3}(s) \rig] \right| ds + \tilde \xi(0),
\endequarrs
for some $\ell_4>0$, where we have used the same argument invoked above to get the second bound. Because of the exponential convergence to zero of its arguments, the integral above converges to a constant as $t \to \infty$, consequently, we can write
\begequ
\lab{x4xi}
\mathbf{x}_{4}(t)= \xi(t) + \theta + \et,
\endequ
for some constant vector $\theta$---equation \eqref{x4xi} corresponds to the key relationship \eqref{keyrel} of PEBO with $\phi(x)=\mathbf{x}_4$.

To complete the proof we show now that, under the persistent excitation condition \eqref{pe}, the proposed estimator is consistent, that is, $\lim_{t\to\infty}\hat \theta(t)=\theta$ that, together with \eqref{x4} and  \eqref{x4xi} establishes the claim that $\tilde {\mathbf x}_4(t) \to 0$. Towards this end, notice that replacing \eqref{sk} in the $k$-th equation of $\dot {\mathbf x}_1$ we get
\begequarrs
\dot {\mathbf x}_{1,k}  = {f}_{1,k}( \mathbf{x}_1,\mathbf{x}_{2},\mathbf{x}_{3},u)  + b^\top (\mathbf{x}_1,\mathbf{x}_{2},\mathbf{x}_{3},u) (\xi + \theta),
\endequarrs
where we have use \eqref{x4xi} to get the second equation. On the other hand, $\dot {y}_k= \dot {\mathbf x}_{1,k}$, hence applying the filter ${\alpha \over p + \alpha}$ we get the (ideal) regression form $Y=\psi^\top \theta$ with $Y  :=  {\alpha p \over p + \alpha} \big[{y}_k\big] -  {\alpha  \over p + \alpha}\big[{f}_{1,k}(y,{ \mathbf{x}}_{2},{ \mathbf{x}}_{3},u) \big]  -  {\alpha  \over p + \alpha}\big[b^\top (y,{ \mathbf{x}}_{2},{ \mathbf{x}}_{3},u)\xi \big],\;
\psi :=  {\alpha  \over p + \alpha}\big[b (y,{ \mathbf{x}}_{2},{ \mathbf{x}}_{3},u) \big],
$ that is, of course,  unmeasurable because of the dependence of ${f}_{1,k}$ and $b$ on the unknown states.  However, due to the fact that the estimation errors $\tilde{ \mathbf x}_2(t)$ and $\tilde{ \mathbf x}_3(t)$ converge exponentially fast to zero, we have that $\hat Y(t) = Y(t) + \et$ and $\hat \psi(t)=\psi + \et$.  Therefore, neglecting the terms $\et$, we get $ \hat{Y} =\hat {\psi}^\top {\theta}$. Replacing the equation above in  \eqref{paa} we get the parameter estimation error equation
$$
\dot{ \tilde{\theta}}  =   \Gamma \hat \psi \hat {\psi}^\top\tilde{\theta},
$$
where $ \tilde{\theta}:= \hat{\theta}- {\theta}$. The proof of (exponential) convergence of $\tilde \theta(t)$ to zero is completed invoking standard adaptive control arguments.
\end{proof}

For the sake of clarity we have presented Proposition \ref{pro7} in a very simple form, being possible to extend it in several directions.
\begite
\item Clearly, the number of subsystems of the form $\dot{\mathbf{x}}_{i}  =  {\mathbf A}_{i} {\mathbf{x}}_{i} + \mathbf{f}_{i} (y, {\mathbf{x}}_{1}, \cdots, {\mathbf{x}}_{n-1},u)$ can be larger than the two taken here.
\item Invoking the recent results of identification and adaptive control of nonlinearly parameterised systems---see \cite{LIUetal} and references therein---it is possible to replace Assumption (i) by:\\

(i') There exists $1 \leq k \leq n_1$ such that the corresponding element of the vector $S$ satisfies
$$
S_k({x},u)=b^\top (\mathbf{x}_1,\mathbf{x}_{2},\mathbf{x}_{3},u)\Phi( \mathbf{x}_{4}),
$$
for some {\em monotonic} mapping $\Phi:\rea^{n_4} \to \rea^{n_4}$.

\item Regarding  Assumption (i) it is also possible to consider the existence, not just of one element of $S$, but several of them verifying the factorizability condition. This will give rise to a matrix regressor $b$ for which the persistent excitation condition \eqref{pe} would be easier to satisfy.

\item For simplicity the unknown parameter $\theta$ is identified in Proposition \ref{pro7} with the classical gradient estimator \eqref{paa}. However, it is possible to replace this estimator with the high-performance dynamic regressor extension and mixing proposed in \cite{ARAetaltac}, see also  \cite{ortega2017sub}. As shown in these papers parameter convergence is ensured without the, often restrictive, persistent excitation condition \eqref{pe}.
\endite

\subsection{DC-DC \'Cuk converter}
\lab{subsec53}
%
In this section, we consider the widely studied DC-DC \'Cuk converter, depicted in Fig. \ref{fig2}, for which a PEBO and an I\&IO were reported in  \cite{ortega2015scl} and  \cite{astolfi2008book}, respectively. We also design a KKLO, a [KKL+PEB]O  and two high gain observers (HGOs) {\em \`a la} \cite{esfandiari1992ijc}. The purpose of this example is to compare, via simulations, the performance of all these observers from the point of view of gain tuning flexibility and robustness with respect to measurement noise, which is unavoidable in this application.

\begin{figure}
  \centering
  \includegraphics[width=9
  cm]{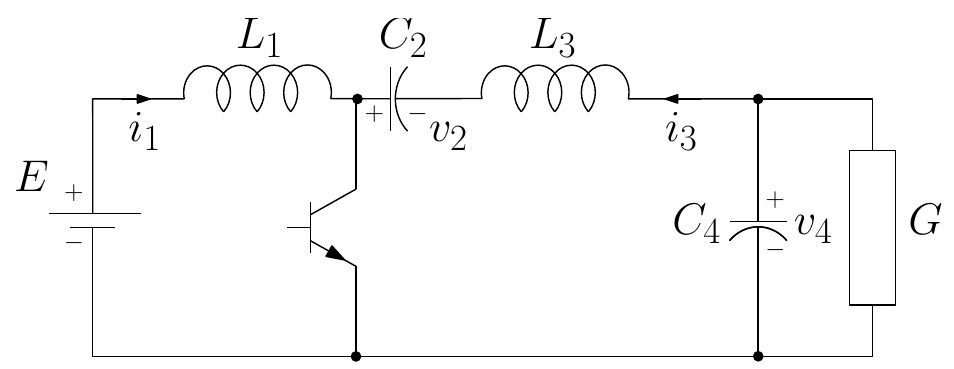}
  \caption{DC-DC \'Cuk converter circuit}\label{fig2}
\end{figure}

The averaged model of the system is given as
\begequ
\begin{aligned}
  \dot{x}_1 & =  - \frac{1}{L_1} (1-u) y_1 + \frac{E}{L_1} \\
  \dot{x}_2 & =   \frac{1}{C_4} y_2 - \frac{G}{C_4}x_2 \\
  \dot{y}_1 & =  \frac{1}{C_2} (1-u) x_1 + \frac{1}{C_2}u y_2 \\
  \dot{y}_2 & =  - \frac{1}{L_3} u y_1 - \frac{1}{L_3} x_2,
\end{aligned}
\endequ
where $x:= \col(i_1,v_4)$, $y: = \col(v_2,i_3)$, and $L_1,C_2,L_3,C_4,E,G$ are positive constants. $u\in(0,1)$ is a duty cycle. We are interested in estimating $x$ with $y$ measurable. Following the observer designs proposed in this note and the ones reported in the literature, we obtain the observers given in Table \ref{tab1}, in which $F(p)={\alpha \over p+\alpha}$ and $W(p)={\alpha p \over p+\alpha}$. Notice that for the KKLO, $\Lambda$ is a time-varying stable matrix, since $1 -u \notin \mathcal{L}_2$.


\begin{table*}[]
\centering
\caption{\rm State observers for C\'uk converters}
\label{tab1}
\begin{tabular}{l|l|l}
\hline
  Type & Observer structure & Mappings\\
\hline
\begin{minipage}{2.5cm}
   KKLO
\end{minipage}&
\begin{minipage}{5cm}
$~~~\eqref{kklo}, \;\hat{x}=\col(L_1^{-1}\xi_2 +  L_1^{-1} C_2 y_1, \xi_1)$
\end{minipage}
&
\begin{minipage}{6cm}
$$
\begin{aligned}
\Lambda & = \diag (-C_4^{-1}{G} , - L_1^{-1}{(1-u)})\\
B              & = \col( C_4^{-1} y_2 ,   ( 1+ L_1^{-1}C_2) (-1 + u) y_1 + E - uy_2 )
\end{aligned}
$$
\end{minipage}
      \\
    \hline
\begin{minipage}{2cm}
    PEBO \cite{ortega2015scl}
\end{minipage}
&
\begin{minipage}{2cm}
 $$
\begin{aligned}
\text{ ~\eqref{dotxi},~~}
 \dot{\hat{\theta}}    & =  \Gamma \mathbf{M}^\top(\mathbf{Y} - \mathbf{M} \hat{\theta}) \\
 \hat{x}                      & = \hat{\theta} + \xi +  \col(0,C_4^{-1}{GL_3} y_2)
\end{aligned}
$$
\end{minipage}
&
\begin{minipage}{7cm}
$$
\begin{aligned}
B & =  \col(L_1^{-1}(E-(1-u)y_1), C_4^{-1}(y_2+ G uy_1))\\
\mathbf{M} &  = \diag(C_2^{-1}F [1-u], - L_3^{-1})
 \\
 \mathbf{Y} &  =
 \begin{bmatrix}
  W [ y_1]  - C_2^{-1} F[ \xi_1(1-u) + uy_2]  \\
  W [y_2]  +  F[  L_3^{-1}(uy_1 + \xi_2 )+ C_4^{-1}(GL_3 y_2)]
 \end{bmatrix}
 ,\;
 \alpha ,\Gamma \in \mathbb{R}_+\\
 \end{aligned}
$$
\end{minipage}
    \\
\hline
\begin{minipage}{2cm}
   [KKL+PEB]O
\end{minipage}&
\begin{minipage}{3cm}
$$
\begin{aligned}
     \eqref{pebo-eq},~~    \dot{\hat{\theta}} & =  \gamma M (Y - M \hat{\theta})\\
                \hat{x}                   & = \col(\xi_1 + \hat{\theta},  \xi_2)
 \end{aligned}
$$
\end{minipage}
&
\begin{minipage}{6cm}
$$
\begin{aligned}
\Lambda & = \diag(0,  -C_4^{-1}G),\;
             B  = \col(L_1^{-1}({E -(1-u)y_1}),C_4^{-1} y_2),\; P=I \\
            Y & = W[y_1] - C_2^{-1}F[ {(1-u)\xi_2 + uy_2 }] ,\;
            M  = C_2^{-1}F[1-u ],
\;
\alpha , \gamma\in\mathbb{R}_+
\end{aligned}
$$
\end{minipage}
\\
    \hline
\begin{minipage}{2cm}
    I\&IO \cite{astolfi2008book}
\end{minipage}
&\multicolumn{2}{c}{
\begin{minipage}{14cm}
$$
\begin{aligned}
\dot{\xi}_1 & =  - \gamma_1 (1-u)(\xi_1 + C_2 \gamma_1 y_1)  + \gamma_1 uy_2    + L_1^{-1}{(E -(1-u)y_1 )}  \\
\dot{\xi}_2 & =C_4^{-1}({y_2 - G(\xi_2 - L_3\gamma_2 y_2)} ) - \gamma_2 ( uy_1 + \xi_2 - L_3\gamma_2y_2) \\
\end{aligned},\;
\hat{x}       = \xi +\begin{bmatrix} C_2 \gamma_1 y_1 \\  L_3 \gamma_2 y_2 \end{bmatrix}
,\;
\gamma_1,\gamma_2 \in \mathbb{R}_+
\quad\quad\quad\quad\;
$$
\end{minipage}}
%
    \\
\hline
\begin{minipage}{2cm}
    HGO \\
    (time-varying dynamics)
\end{minipage}
&\multicolumn{2}{c}{
\begin{minipage}{14cm}
$$
\dot \xi    =
                \begin{bmatrix}
                C_2^{-1}( 1-u ) \xi_2   + C_2^{-1}(u y_2) + \alpha_1 r_1^{-1} (y_1-\xi_1) \\
                - L_1^{-1}(1-u )  y_1 + L_1^{-1}E +  \alpha_2 r_1^{-2} (y_1-\xi_1) \\
                 - L_3^{-1} \xi_4  - L_3^{-1} u y_1 + \alpha_3 r_1^{-1}(y_2 -\xi_3)\\
                 C_4^{-1} y_2 - C_4^{-1} G \xi_4+ \alpha_4 r_1^{-2}(y_2 -\xi_3) \\
                \end{bmatrix} ,
\; \hat{x} = \begin{bmatrix} \xi_2 \\ \xi_4 \end{bmatrix},\; r_1 \in (0,1], \; \alpha_i > 0
\quad\quad\quad\quad\quad\quad\quad\quad\;\;
$$
\end{minipage}}

\\

\hline
\begin{minipage}{2cm}
    HGO \\
    (linear dynamics)
\end{minipage}
&\multicolumn{2}{c}{
\begin{minipage}{15cm}
$$
\dot \xi     =
                \begin{bmatrix}
                \xi_2 + \alpha_1 r_1^{-1} (y_1-\xi_1) \\
                -  L^{-1}(1-u) y_1 +  L_1^{-1}E  +  \alpha_2 r_1^{-2} (y_1-\xi_1) \quad\\
                  \xi_4 + \alpha_3 r_1^{-1}(y_2 -\xi_3)   \\
                C_4^{-1} y_2 - C_4^{-1}G\xi_4 + \alpha_4 r_1^{-2}(y_2 -\xi_3)\\
                \end{bmatrix}
,\;
\hat{x} =
\begin{bmatrix}
  {C_2 \xi_2- uy_2 \over 1-u}\\
                 - L_3\xi_4 - uy_1
\end{bmatrix},
\; r_2 \in (0,1], \; \alpha_i > 0
\quad\quad\quad\;\;
$$
\end{minipage}}
\\
  \hline
\end{tabular}
\end{table*}

Simulations were conduct with measurement noises, which are generated by Matlab/Simulink's uniform random number block with sampling time of 0.0001s, and the magnitude limitations are [-0.02,0.02] for $y_1$ and [$-2\times 10^{-4},2\times 10^{-4}$] for $y_2$. The parameters of the converter are $L_1=10$ mH, $C_2$=22.0 $\mu$F, $C_4$=22.9 $\mu$F, G=0.0447 S and $E$=12 v. In order to give a fair comparison study, the system runs with the \emph{ideal} state-feedback with the stabilizing control law given in  \cite{astolfi2008book}
$$
u = \frac{|V_d|}{|V_d| + E} + \lambda \frac{G|V_d|v_2 + E(x_2 -x_1)}{1 + (G|V_d|v_2 + E(x_2-x_1))^2},
$$
where $V_d$ is the set point for the output voltage $v_4$, which was is selected as in \cite{ortega2015scl}. The observer parameters were taken as $\alpha=0.5,\;\gamma = 0.001, \;\Gamma=\diag(0.001,100), \;\gamma_1=50,\;\gamma_2=1,\; r_1=0.05, \;r_2=0.005, \; \alpha_1=\alpha_3=2$,$\;\alpha_2=\alpha_4=1$, to make the observers have approximate convergence speeds. All the initial values of the dynamic extensions in observers are selected as $0$. The simulation results are given in Fig. \ref{Doc3}.

The following remarks are in order.
\begin{itemize}
  \item KKLO and I\&IO have two-order dynamics, clearly, the lowest order ones. KLLO has the simplest observer structure. The parameters in PEBO were the easiest to tune with guaranteed convergence speed; KKLO and [KKL+PEB]O need to resolve PDEs to tune. Besides, for HGO the achievable convergence speed is severely limited.
\item The I\&I framework allows to treat in a unified manner the problems of \emph{state and parameter} estimation, see \cite{astolfi2008book} for the state observation with unknown parameters.
  \item The KKLO has the best performance in the presence of measurement noise, probably due to the fact that its dynamic extension is a linear system that attenuates the effect of the noise. On the other hand, the dynamics in PEBO, [KKL+PEB]O and I\&IO are nonlinear, and seem to have a deleterious impact on the noise.
  \item The first HGO yields a \emph{time-varying} error dynamics, which is stable because of the physical constraint $1-u>0$. It has oscillations in the transient stage. The second HGO has LTI dynamics, where high gain injections are used to estimate the output derivatives, \emph{i.e.}, $\dot{y}_1$ and $\dot{y}_2$. The operating modes of the converter switch at the moments $t=0.2k$ s ($k=1,\ldots,5$), yielding relatively \emph{large} derivatives of the outputs around these moments.
  \item As expected, the worst performance was systematically observed for the HGOs because of the high-gain injection needed to ensure its stability. It is worth pointing out that this (well-known) deleterious effect of high-gain injection was also observed for mechanical systems in \cite{ORTetalijc}.
\end{itemize}


\begin{figure*}[h]
  \centering
 \includegraphics[width=7.3cm,height=7cm]{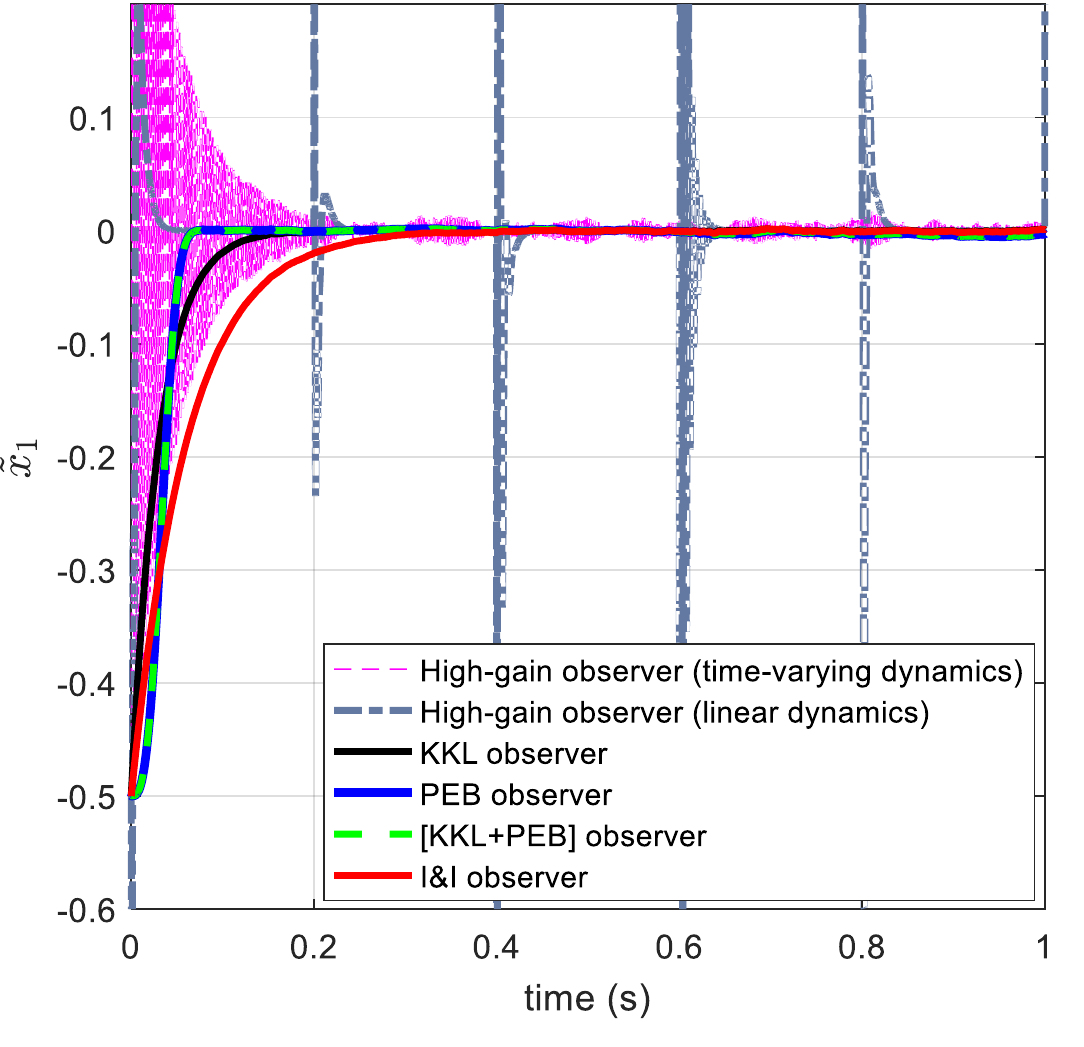}
  \includegraphics[width=7.3cm,height=7cm]{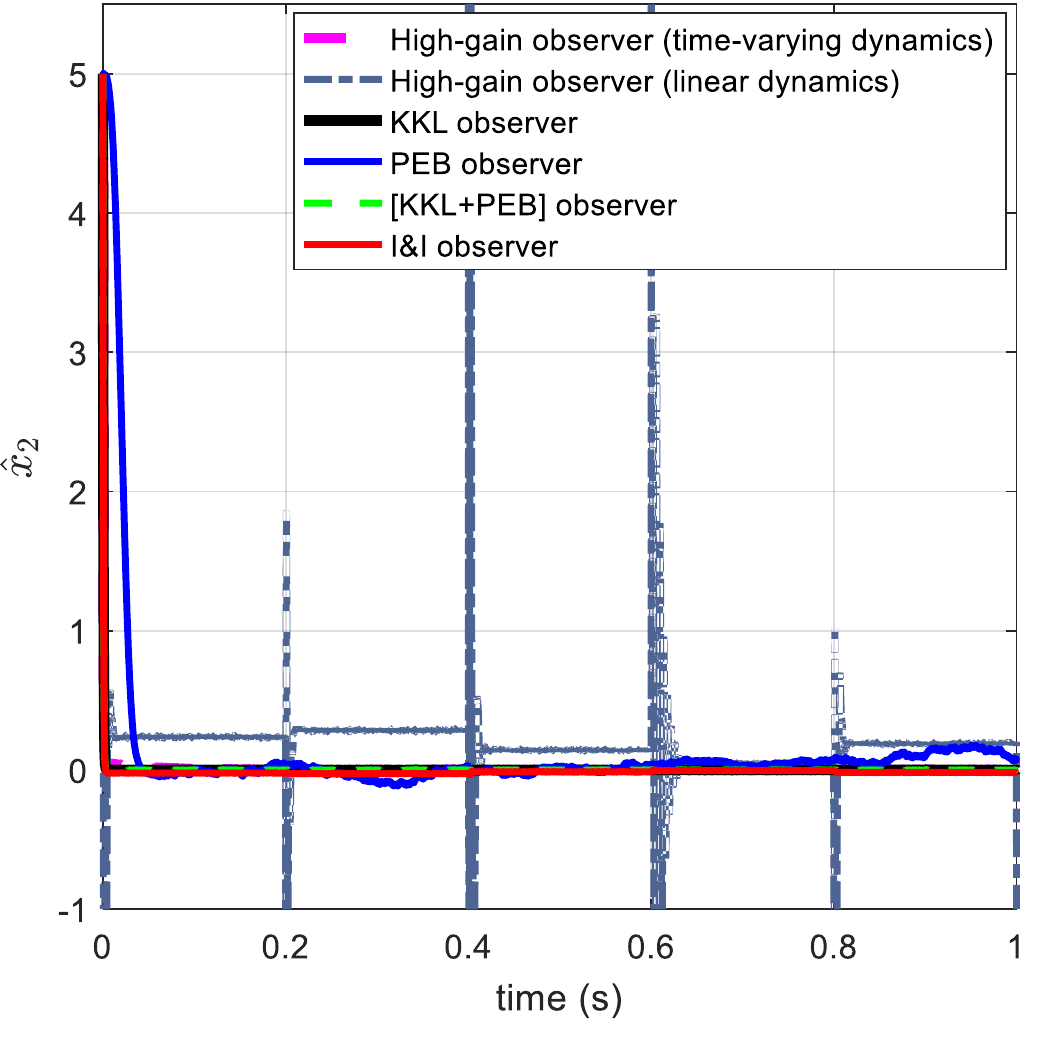}
 \includegraphics[width=7.3cm,height=7cm]{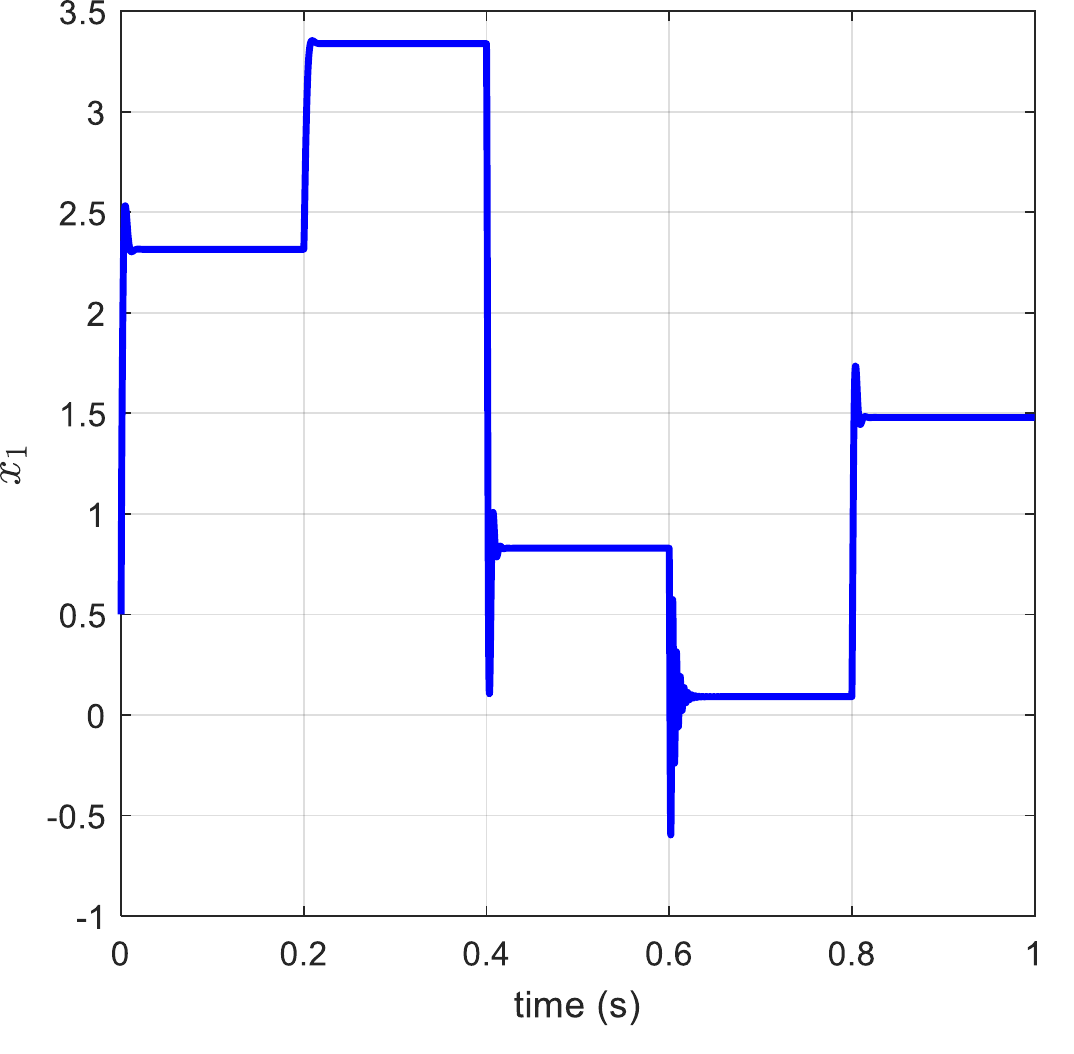}
  \includegraphics[width=7.3cm,height=7cm]{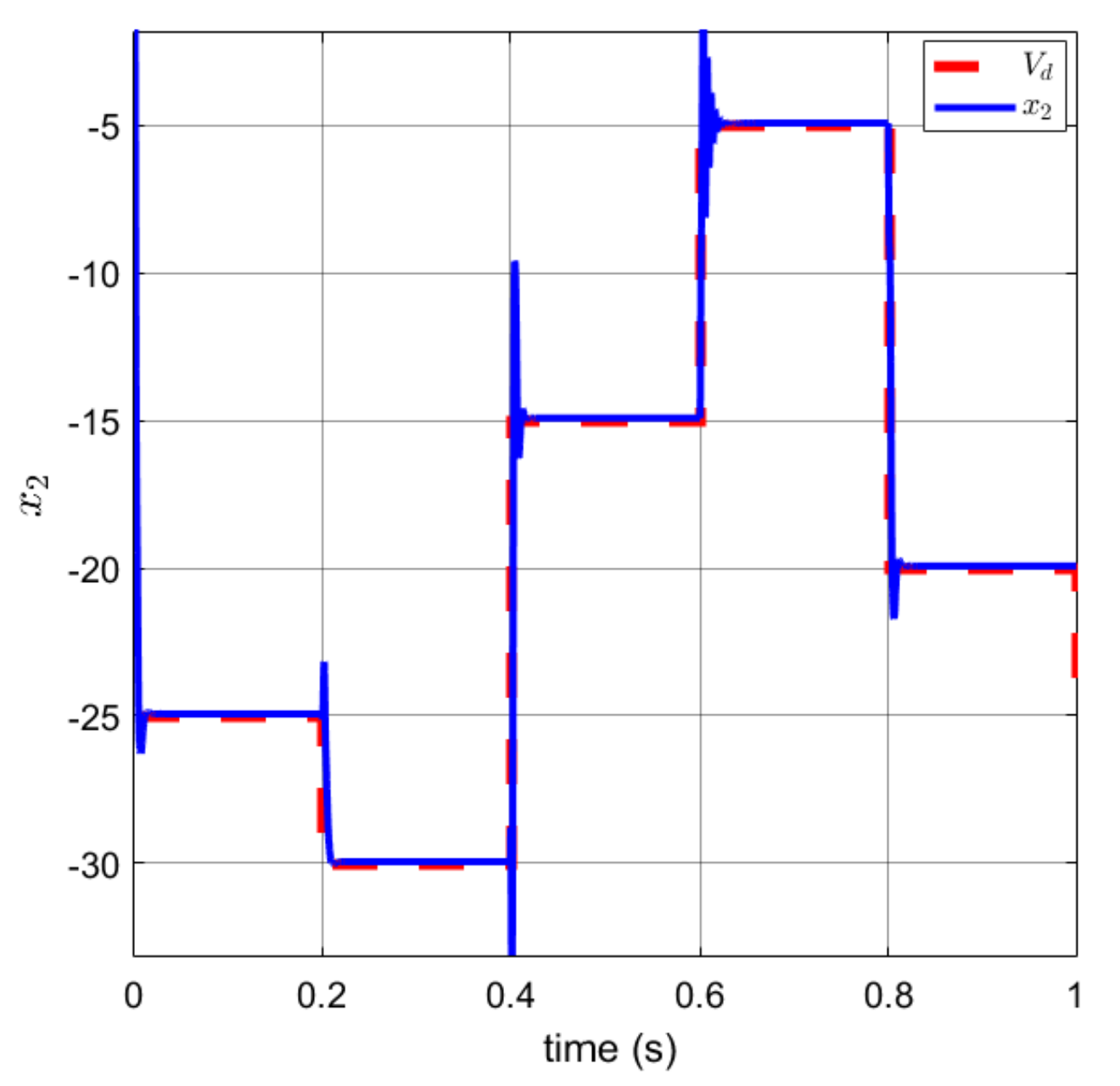}
  \caption{The system state variables $x_1$ and $x_2$ with the full-state feedback controller, and observation errors $\tilde{x}_1 = \hat{x}_1 - x_1$ and $\tilde{x}_2 = \hat{x}_2 - x_2$.}
  \label{Doc3}
\end{figure*}

%
\section{Concluding Remarks}
\label{sec6}
%
A new observer design technique, called [KKL+PEB]O, which consists of the combination of KKLO and PEBO was introduced---providing more degrees of freedom for the solution of the key PDE. Via the suitable selection of the tuning matrix $\Lambda$ of the form \eqref{lam0}, in the PDE \eqref{KKLOPDE}, [KKL+PEB]O reduces to PEBO or KKLO. An example that is not solvable with KKLO nor PEBO, but it is via [KKL+PEB]O show that the new observer design extends the applicability of PEBO and KKLO. Also we  identified a class of nonlinear systems, for which [KKL+PEB]O provides a simple constructive solution.

An additional contribution is the proof that, a slight generalisation of the I\&IO, allows us to obtain [KKL+PEB]O, as well as PEBO and KKLO, as particular cases of  I\&IO. This provides a unified framework, based on immersion and invariance, to treat the three observer designs and establish the ``set" relationship \eqref{setdia}.

Further research is underway in the following directions.
\begin{itemize}
  \item Exploit the constructive approach to find the free mappings in [KKL+PEB]O for some more specific classes of physical systems.
  \item Generalize the coordinate change from $\phi(x)$  to $\phi(x,u)$ in order to simplify the solution of the PDEs in these observers. Along this line of research, one interesting possibility is extending the theoretical observer existence results in \cite{Andrieu2006SIAM} to \emph{control} systems with input.
  \item Study [KKL+PEB]O-based output feedback control. In particular, we are currently investigating if, due to the presence of the PEBO part, the resulting controller enjoys a ``self-tuning" property similar to model reference adaptive control. That is, if the control objective can be achieved without requiring that the parameter estimation error $\tilde{\theta}$ converges to zero. Such a property would obviate the need of excitation conditions  for  [KKL+PEB]O-based (or PEBO-based) output feedback control.
\end{itemize}

\section*{acknowledgement}
The authors would like to thank Stanislav Aranovskiy, Hassan K. Khalil and Laurent Praly for some suggestions and comments in Section \ref{sec5}, also thank Pauline Bernard and Vincent Andrieu for clarifications on detectability in KKL observer design. They are also grateful to the Editor, Associate Editor and anonymous reviewers for their highly thorough revisions.

\end{document}